\documentclass [3p, 12pt, sort&compress]{elsarticle}
\usepackage[numbers]{natbib}
\usepackage{graphicx}
\usepackage[small]{subfigure,epsfig}
\usepackage{indentfirst}
\usepackage{amsmath,latexsym,enumerate}
\usepackage{amssymb}
\usepackage{amsfonts}

\usepackage{amsthm}

\theoremstyle{plain} \newtheorem{theorem}{Theorem}
\newtheorem{lemma}{Lemma}
\numberwithin{equation}{section} \numberwithin{lemma}{section} \numberwithin{theorem}{section}

\makeatletter
\def\ps@pprintTitle{%
 \let\@oddhead\@empty
 \let\@evenhead\@empty
 \def\@oddfoot{\centerline{\thepage}}%
 \let\@evenfoot\@oddfoot}
\makeatother

\begin{document}
\begin{frontmatter}

\title{From Puiseux series to   invariant algebraic curves: the FitzHugh--Nagumo model}

\author{Maria V. Demina}

\address[1]{National Research Nuclear University "MEPhI", 31 Kashirskoe Shosse,
115409, Moscow, Russian Federation}

\address[2]{National Research University Higher School of Economics, 34 Tallinskaya Street, 123458, Moscow, Russian Federation}

\ead{maria\underline{ }dem@mail.ru; mvdemina@mephi.ru}

\begin{abstract}
A relationship between Puiseux series satisfying an ordinary differential equation corresponding to a polynomial dynamical system  and degrees of irreducible invariant algebraic curves
is studied. A bound on the degrees of irreducible invariant algebraic curves  for a wide class of  polynomial dynamical systems is obtained. It is demonstrated that  the Puiseux series near infinity can be used to find irreducible algebraic curves explicitly. As an example, all irreducible invariant algebraic curves for  the famous FitzHugh--Nagumo system  are obtained.

\end{abstract}

\begin{keyword}
FitzHugh--Nagumo model, invariant algebraic curves, Darboux polynomials, Puiseux series

\end{keyword}

\end{frontmatter}

\section{Introduction}\label{Introduction}

A polynomial dynamical system in $\mathbb{C}^2$ can be defined as
\begin{equation}
\begin{gathered}
 \label{VF}
 x_t=P(x,y),\quad y_t=Q(x,y),
\end{gathered}
\end{equation}
where $P(x,y)$ and $Q(x,y)$ are polynomials in the ring $\mathbb{C}[x,y]$. An algebraic curve $F(x,y)=0$, $F(x,y)\in \mathbb{C}[x,y]\setminus\mathbb{C}$
is called an invariant algebraic curve  (or a Darboux polynomial) of dynamical system \eqref{VF} if it satisfies the following equation
\begin{equation}
\begin{gathered}
 \label{Inx_Eq}
P(x,y)F_x+Q(x,y)F_y=\lambda(x,y) F,
\end{gathered}
\end{equation}
where $\lambda(x,y)\in \mathbb{C}[x,y]$ is  a polynomial called the cofactor of the invariant curve $F(x,y)=0$. For simplicity let us slightly abuse notation and call the polynomial $F(x,y)$ satisfying equation
 \eqref{Inx_Eq}  an invariant algebraic curve bearing in mind that in fact the zero set of $F(x,y)$ is under consideration.

\begin{lemma}\label{L1}

Let $F(x,y)\in \mathbb{C}[x,y]\setminus\mathbb{C}$ and $F=f_1^{n_1}\ldots f_m^{n_m}$ be its factorization in irreducible factors.
Then $F$ is an invariant algebraic curve of dynamical system \eqref{VF}  if and only if $f_j$ is an
invariant algebraic curve  of  dynamical system \eqref{VF} for each $j = 1$, $\ldots$, $m$ . In addition the following relation is valid $\lambda = n_1\lambda_1+$ $\ldots$ $+n_m\lambda_m$, where $\lambda$ is the cofactor of $F$ and $\lambda_j$ is the cofactor of $f_j$.

\end{lemma}

The proof of this lemma is straightforward, see, for example, \cite{Llibre01}. It can be observed that an invariant
algebraic curve of dynamical system \eqref{VF} is formed by solutions of the latter. A solution of dynamical system \eqref{VF}
has either empty intersection with the zero set of $F$, or it is entirely contained in $F= 0$.
Existence of invariant algebraic curves is a substantial measure of integrability, for more details see \cite{Singer, Llibre01, Christopher, Goriely, Zhang, Lei01}.
In view of Lemma \ref{L1} it is an important problem to classify all irreducible invariant algebraic curves of
dynamical systems.

The problem of finding a bound on degrees of irreducible invariant algebraic curves goes back to Poincar\'{e}~\cite{Poincare}.
Partial results valid under certain restrictions on the singularities of dynamical systems or (and) invariant algebraic curves were obtained
by Cerveau and Lins Neto \cite{Neto01}, Carnicer \cite{Carnicer01},  Walcher \cite{Walcher01}.
It is still an open problem to establish an "effective" upper bound (if any) for a given polynomial dynamical system. Here by "effective" upper bound we mean a bound, which allows one to  find all irreducible invariant algebraic curves iterating finite amount of times  the method of undetermined coefficients. For rapid methods of finding Darboux polynomials with bounded degrees see \cite{Cheze01}. Note that the
uniform upper bound that depends entirely on the degree of the dynamical system under consideration may not
exist. Indeed, there are polynomial dynamical systems that possess irreducible invariant algebraic curves with degrees depending on the coefficients of the system. For example, if the coefficients in a family of quadratic dynamical systems studied by Christopher and Llibre   are not bounded, then the degrees of irreducible invariant algebraic curves can be  arbitrary \cite{Llibre03}.

In this article our aim is to present an  approach, which can be used to find
an "effective" bound for a wide class of  polynomial dynamical systems. Our main tool is to use Puiseux series satisfying the following first--order ordinary differential equation
\begin{equation}
\begin{gathered}
 \label{ODE_y}
P(x,y)y_x-Q(x,y)=0.
\end{gathered}
\end{equation}
Here the variable $y$ is regarded as dependent and the variable $x$ as independent. Note that the roles can be changed.
In what follows we shall suppose that  the polynomials $P(x,y)$, $Q(x,y)$ do not have non--constant common factors.

In relationship with our approach let us also mention the work by Lei and Yang \cite{Lei02}, where the bounds for irreducible invariant algebraic curves of certain dynamical systems were obtained in terms of algebraic multiplicities of the dynamical systems in question at the singular points, and   the method for finding rational first integrals
of two--dimensional polynomial vector fields introduced by Ferragut and Giacomini \cite{Ferragut01}. This method uses Tailor and Puiseux series
near finite points. In addition see the articles also dealing with first integrals, algebraic functions, and Puiseux series near finite points
\cite{Lei03, Giacomini02, Giacomini01, Gine01}.

It seams that the connection between  invariant
algebraic curves and  Puiseux series near infinity that satisfy  first--order ordinary differential equation  \eqref{ODE_y}  is discussed   in this article for the first time.

A Puiseux series in a neighborhood of the point $x=0$ is defined as
\begin{equation}
\begin{gathered}
 \label{Puiseux_null}
y(x)=\sum_{k=0}^{+\infty}c_kx^{\frac{l_0}{n_0}+\frac{k}{n}}
\end{gathered}
\end{equation}
where $l_0\in\mathbb{Z}$, $n_0\in\mathbb{N}$.

In its turn a Puiseux series in a neighborhood of the point $x=\infty$ is given by
\begin{equation}
\begin{gathered}
 \label{Puiseux_inf}
y(x)=\sum_{k=0}^{+\infty}b_kx^{\frac{l_0}{n_0}-\frac{k}{n}}
\end{gathered}
\end{equation}
where $l_0\in\mathbb{Z}$, $n_0\in\mathbb{N}$.

Let us formulate our main results.

\begin{theorem}\label{T1}
Suppose that there exists finite number of Puiseux series of the form \eqref{Puiseux_null}
satisfying equation \eqref{ODE_y}. Let $F(x,y)$ be an irreducible invariant  algebraic curve of dynamical system \eqref{VF}.
Then  the degree of $F(x,y)$ with respect to $y$  does not exceed the number of distinct Puiseux series of the form \eqref{Puiseux_null} that satisfy equation \eqref{ODE_y}.
\end{theorem}

Situation of Theorem \ref{T1} seldom occurs. The next theorem is much more important for applications.

\begin{theorem}\label{T2}
Suppose that there exists finite number of Puiseux series of the form \eqref{Puiseux_inf}
satisfying equation \eqref{ODE_y}. Let $F(x,y)$ be an irreducible invariant  algebraic curve of dynamical system \eqref{VF}.
Then  the degree of $F(x,y)$ with respect to $y$  does not exceed the number of distinct Puiseux series of the form \eqref{Puiseux_inf} that satisfy equation \eqref{ODE_y}.
\end{theorem}

Interestingly, the same "finiteness property" of Laurents series is one of the major points in classification of meromorphic solutions of autonomous algebraic ordinary differential equations~\cite{Eremenko01, Demina02, Demina04, Demina03}.

All the Puiseux series in neighborhoods of the points $x=0$ and $x=\infty$ satisfying equation \eqref{ODE_y} can be easily constructed with the help of the Newton polygon related to equation \eqref{ODE_y} \cite{Bruno02, Bruno01}. Note that  the Newton diagram and the Puiseux series in a neighborhood of the point $x=0$ are mainly considered in classical literature. In view of this  we shall give a definition of the Newton polygon and describe  an algorithm for finding the Puiseux series for booth cases in detail.

Let us note that making the change of variables  $x\mapsto x-x_0$, $x_0\neq0$ or $y\mapsto y-y_0$, $y_0\neq0$ one can investigate the number of admissible Puiseux series in a neighborhood of the point $x=x_0$ or $y=y_0$. A theorem similar to \ref{T1} can be formulated for    Puiseux series centered at the point $x=x_0$ or $y=y_0$.
Here and in what follows an admissible Puiseux series means that this series  satisfies equation \eqref{ODE_y}.


In this article we shall consider the following  dynamical system
\begin{equation}
\begin{gathered}
 \label{FH1_DS}
x_t=-x^3+ex^2+\sigma x-y+\delta,\quad y_t=\alpha x+\beta y.
\end{gathered}
\end{equation}
The change of variables $x\mapsto x+A$, $y\mapsto y+B$, $A=e/3$, $\alpha A+\beta B=0$ relates system \eqref{FH1_DS} with its simplified version at $e=0$. Thus without loss of generality we set $e=0$.
All the parameters are supposed to be from the field $\mathbb{C}$. If $\alpha=0$ and $\beta=0$, then integrating the second equation, we obtain $y(t)=C_0$. In this case the function $x(t)$  satisfies a simple first--order ordinary differential equation. Hence in what follows we suppose that $\alpha$ and $\beta$ are not simultaneously zero.  Dynamical system  \eqref{FH1_DS}  is the two--dimensional FitzHugh--Nagumo system \cite{FitzHugh01, Nagumo01}. It is one of the most famous models describing the excitation of neural membranes and
the propagation of nerve impulses along an axon. This system has been intensively studied in recent years, see \cite{Kostova, Langfield01, Demina11} and references therein. In article \cite{Llibre02} first--order (with respect to $y$) invariant algebraic curves were derived. In this article our aim is to obtain all irreducible algebraic   curves for  dynamical system~\eqref{FH1_DS}.

\begin{theorem}\label{T:FH_DP}

The unique irreducible invariant algebraic curves of dynamical system \eqref{FH1_DS} with $e=0$ and $|\alpha|+|\beta |>0$ are those given in table \ref{T:FH_inv}.

\end{theorem}

It seems that the classification of irreducible invariant algebraic curves for the FitzHugh--Nagumo system is given here for the first time.

This article is organized as follows. In section  \ref{P} we prove our main results and also present other theorems for the case of infinite number of admissible Puiseux series.
In section \ref{Newton} we give a definition of the Newton polygon of an algebraic ordinary differential equation and describe a method, which can be used to construct series representing solutions of this equation in neighborhoods of the origin and infinity.  In section  \ref{FH_Darboux} we classify irreducible invariant algebraic curves of the FitzHugh--Nagumo system.

\section{Proof of main results} \label{P}

It is known \cite{Walker} that a Puiseux series of the form \eqref{Puiseux_null} that satisfies equation $F(x,y)=0$ is convergent in a neighborhood of the point $x=0$ (the point $x=0$ is excluded from domain of convergence if $l_0<0$). Analogously, a Puiseux series of the form \eqref{Puiseux_inf} that satisfies equation $F(x,y)=0$ is convergent in a neighborhood of the point $x=\infty$ (the point $x=\infty$ is excluded from domain of convergence if $l_0>0$). This fact follows from the classical result if we consider the change of variables $s=x^{-1}$, which brings infinity to the origin. In other words the   Puiseux series $y(s^{-1})$ satisfying the equation $G(s,y)=0$  converges in a neighborhood of the point $s=0$ (the point $s=0$ is excluded from domain of convergence if $l_0>0$). Here $G(s,y)=s^MF(s^{-1},y)\in\mathbb{C}[s,y]$ is an algebraic curve, $M$ is the degree of $F(x,y)$ with respect to $x$. The set of all Puiseux series of the form \eqref{Puiseux_null} (\eqref{Puiseux_inf}) forms a field, which we denote by $\mathbb{C}\{x\}$ ($\mathbb{C}_{\infty}\{x\}$). For more details see \cite[chapter IV, paragraph 3]{Walker}.

Let us prove the following lemma.

\begin{lemma}\label{L2}

Let $y(x)$ be  a Puiseux series satisfying the equation $F(x,y)=0$, $F_y\not\equiv0$ with $F(x,y)$ being an invariant algebraic curve    of dynamical system \eqref{VF}.
 Then the series $y(x)$ satisfies equation  \eqref{ODE_y}.
\end{lemma}

\begin{proof}
Representing $F(x,y)$ as the product of irreducible factors $F=f_1^{n_1}\ldots f_m^{n_m}$, we see that there exists $f_j(x,y)$ such that $f_j(x,y(x))=0$.
Differentiating this equation with respect to $x$, we get
\begin{equation}
\begin{gathered}
 \label{L2_1}
f_{j,\,x}(x,y(x))+y_xf_{j,\,y}(x,y(x))=0.
\end{gathered}
\end{equation}
It follows from Lemma \ref{L1} that $f_j(x,y)$ is an irreducible invariant algebraic curve    of dynamical system \eqref{VF} and satisfies the equation
\begin{equation}
\begin{gathered}
 \label{Inx_Eq_f}
P(x,y)f_{j,\,x}+Q(x,y)f_{j,\,y}=\lambda_j(x,y) f_j.
\end{gathered}
\end{equation}
 Substituting $y=y(x)$ into this equation yields
\begin{equation}
\begin{gathered}
 \label{Inx_Eq_f1}
P(x,y(x))f_{j,\,x}(x,y(x))+Q(x,y(x))f_{j,\,y}(x,y(x))=0.
\end{gathered}
\end{equation}
Further, let us note that the series $y=y(x)$ cannot satisfy the equation $f_{j,\,y}(x,y)=0$. Indeed, assuming the contrary we see that $f_j$
and $f_{j,\,y}$ intersect  in an infinite number of
points inside the domain of convergence of the series $y=y(x)$. It follows from the  B\'{e}zout's  theorem that there exists a polynomial both
dividing $f_j$ and $f_{j,\,y}$. Since $f_j$ is irreducible, we conclude that this divisor coincides with $f_j$. Thus we get $f_{j,\,y}=fh$ with $h$ being a polynomial.
This relation contradicts the fact that the degree of $f_{j,\,y}$ is less than the degree of $f_j$.

Homogeneous system of linear equations \eqref{L2_1}, \eqref{Inx_Eq_f1} relating $f_{j,\,x}(x,y(x))$, $f_{j,\,y}(x,y(x))$  has non--trivial solutions. Indeed, $f_{j,\,y}(x,y(x))\neq0$ Consequently
its determinant equals  zero. This completes the proof.

\end{proof}

\textit{Proof of Theorem \ref{T1}.} Using the classical Newton--Puiseux algorithm, it can be shown that the field $\mathbb{C}\{x\}$ is algebraically closed
\cite[chapter IV, paragraph 3]{Walker}. Let $F(x,y)$, $F_y\not\equiv0$  be an invariant algebraic curve of polynomial dynamical system \eqref{VF}. There exists uniquely determined system of elements $y_n(x)\in \mathbb{C}\{x\}$ such that the following representation is valid~\cite[chapter IV, paragraph 3, theorem 3.2]{Walker}
\begin{equation}
\begin{gathered}
 \label{F_rep_0}
F(x,y)=\mu(x)\prod_{n=1}^N\{y-y_n(x)\},
\end{gathered}
\end{equation}
where $N$ is the degree of $F(x,y)$ with respect to $y$ and $\mu(x)\in \mathbb{C}[x]$. Moreover, if a non--constant polynomial $g(x)\in \mathbb{C}[x]$ does not divide $F(x,y)$, then $F(x,y)$ has multiple factors in $\mathbb{C}[x,y]$ if and only if the equation $F(x,y)=0$
has multiple roots in $\mathbb{C}\{x\}$ \cite[chapter IV, paragraph 3, theorem 3.5]{Walker}. Further, it follows from Lemma \ref{L2} that the set of elements $y_n(x)\in \mathbb{C}\{x\}$ appearing in representation \eqref{F_rep_0} is a subset of those satisfying equation \eqref{ODE_y}. If the latter is finite and $F(x,y)$ is irreducible, then $N$ does not exceed the number of distinct Puiseux series of the form \eqref{Puiseux_null}
satisfying equation \eqref{ODE_y}. Indeed, if the degree of $F(x,y)$ with respect to $y$ exceeds the number of admissible
distinct Puiseux series in a neighborhood of the point $x=0$, then $F(x,y)$ has multiple roots in representation \eqref{F_rep_0} and consequently   is reducible. This completes the proof.

\textit{Proof of Theorem \ref{T2}.} We repeat the proof of Theorem \ref{T1} with the field  $\mathbb{C}\{x\}$ replaced by $\mathbb{C}_{\infty}\{x\}$. Note that representation of $F(x,y)$ in $\mathbb{C}_{\infty}\{x\}$ reads as
\begin{equation}
\begin{gathered}
 \label{F_rep_inf}
F(x,y)=\mu(x)\prod_{n=1}^N\{y-y_n(x)\},
\end{gathered}
\end{equation}
where $y_n(x)\in \mathbb{C}_{\infty}\{x\}$ and $\mu(x)\in \mathbb{C}[x]$.


Let us prove other theorems, applicable even if equation \eqref{ODE_y} admits infinite number of Puiseux series.

\begin{theorem}\label{T3_0}

Let $y_j(x)\in \mathbb{C}\{x\}$  be  a Puiseux series with uniquely determined coefficients satisfying  equation \eqref{ODE_y}. Then the degree of $y-y_j(x)$ in representation  \eqref{F_rep_0}  of an irreducible invariant algebraic curve  $F(x,y)$  of polynomial dynamical system \eqref{VF} is either $0$ or~$1$.
\end{theorem}

\begin{proof}
Let $F(x,y)$ be  an irreducible invariant algebraic curve    of dynamical system \eqref{VF}. The curve $F(x,y)$ can be represented in the form \eqref{F_rep_0}.  Suppose that  the factor $y-y_j(x)$ appears in this representation  at least twice. Then $F(x,y)$ has multiple factors in $\mathbb{C}\{x\}$  and consequently in $\mathbb{C}[x,y]$ \cite[chapter IV, paragraph 3, theorem 3.5]{Walker}. This fact contradicts irreducibility of $F(x,y)$. The proof is completed.
\end{proof}

\begin{theorem}\label{T3}

Let  $y_j(x)\in \mathbb{C}_{\infty}\{x\}$ be  a Puiseux series with uniquely determined coefficients satisfying  equation \eqref{ODE_y}. Then the degree of $y-y_j(x)$ in representation  \eqref{F_rep_inf} of an irreducible invariant algebraic curve  $F(x,y)$  of dynamical system \eqref{VF} is either $0$ or~$1$.
\end{theorem}

\begin{proof}
We prove this theorem repeating the proof of Theorem \ref{T3_0} with the field  $\mathbb{C}\{x\}$ replaced by $\mathbb{C}_{\infty}\{x\}$.
\end{proof}

Concluding this section let us mention that one can state the same theorems if the variables $x$ and $y$ change their roles or if the base of Puiseux series is the point $x_0\neq0$ and $x_0\neq\infty$.

\section{Newton polygons and  series representing solutions of ordinary differential equations} \label{Newton}

Let us consider an algebraic ordinary differential equation
\begin{equation}
\begin{gathered}
 \label{Alg_ODE}
E\left(\frac{d^k\,y}{dx^k},\ldots,\frac{d\,y}{dx},y,x\right)=0,
\end{gathered}
\end{equation}
where $E$ is a polynomial of its arguments. This equation can be regarded as the sum of differential monomials of the form
 \begin{equation}
\begin{gathered}
 \label{Monomials}
M[y(x),x]=cx^ly^{j_0}\left\{\frac{d y}{dx}\right\}^{j_1}\ldots \left\{\frac{d^k y}{dx^k}\right\}^{j_k},\quad c\in\mathbb{C}.
\end{gathered}
\end{equation}
In what follows we denote  by $W[y(x),x]$ a polynomial in $x$, $y(x)$, and its derivatives. This polynomial is referred to as a differential polynomial.
Note that we are  in the framework of analytic theory of differential equations, i.e. $x$ is a complex variable and $y(x)$ is a complex--valued function.

\textit{Definition.} We say that  algebraic ordinary differential equation \eqref{Alg_ODE} has \textit{a dominant balance} $E_0[y(x),x]$ related to the point $x=\infty$
if the following conditions are valid:
\begin{enumerate}

\item each differential monomial $M[y(x),x]$ appearing in $E_0[y(x),x]$  is also involved in  original equation \eqref{Alg_ODE};

\item the exists a power function $y(x)=b_0x^r$ (possibly, not unique) with
$b_0\neq0$, $r\in\mathbb{C}$ such that all the monomials $M[y(x),x]$ of $E_0[y(x),x]$ have the same exponent
$s\in\mathbb{C}$ in the relation
$M[b_0x^r,x]=C_Mx^s$;

\item for all the monomials $L[y(x),x]$ of  equation \eqref{Alg_ODE} that are not involved in $E_0[y(x),x]$
we obtain $L[b_0x^r,x]=C_Lx^{p_L}$, where $\text{Re}$ $p_L<$ $\text{Re}$ $s$.

\end{enumerate}

The definition of   dominant balances related to the point $x=0$ can be introduced similarly. The only condition we should change
is the  third one. In this case the last inequality in the third item is $\text{Re}$ $p_L>$ $\text{Re}$ $s$.

A direct  and simple way to find all the dominant balances
is to use the Newton polygon of the algebraic ordinary differential equation under consideration. This technique now known as the power geometry was developed by Bruno \cite{Bruno02, Bruno01}.

 Define the map $q:$ $M \rightarrow \mathbb{R}^2$ by  the following rules
\begin{equation}\label{e:2.2}
cx^{q_1}y^{q_2} \mapsto q=(q_1,q_2), \qquad
\frac{d^ky}{dx^k} \mapsto q=(-k,1),\quad q(M_1M_2)=q(M_1)+q(M_2),
\end{equation}
where $c\in \mathbb{C}$ is a constant, $M_1$ and $M_2$ are differential monomials. We denote the set of all points $p\in \mathbb{R}^2$ corresponding to the monomials of equation  \eqref{Alg_ODE} as $S(E)$.

\textit{Definition.} The convex hull of $S(E)$ is called the Newton polygon of
equation   \eqref{Alg_ODE}.

The boundary of the Newton polygon consists of vertices and edges. Selecting  all the differential monomials of the original equation that generate the vertices and the edges of the Newton polygon, we obtain a number of  balances. The functions solving these balances produce asymptotics (at $x\rightarrow 0$ or $x\rightarrow \infty$) of solutions of equation~\eqref{Alg_ODE} \cite{Bruno02, Bruno01}.

In this article we are interested in power asymptotics. Thus we substitute $y=cx^r$, $c\neq0$ into a balance.  If the balance corresponds to an edge that is not parallel to the $q_1$--axis, then $r\in\mathbb{Q}$, $c\in\mathbb{C}$ are fixed.   The parameter $r$ is unique, while the parameter $c$ may take  several distinct values. Balances related to edges parallel to the $q_1$--axis do not have power solutions. If the balance is not algebraic and corresponds to a vertex, then $c\in\mathbb{C}$ is arbitrary (but not zero), the parameter $r$ can  be complex--valued ($r\in\mathbb{C}$) and there may exist  a set of values for $r$. Algebraic balances related to vertices do not have non--trivial solutions.

In what follows we are working in the frames of two-dimensional Euclidean space $\mathbb{R}^2$. All the vectors and rays which are to appear below have the origin $(q_1,q_2)=(0,0)$ as a starting point.
By $\psi$ we denote the angle between the external normal to an edge (external with respect to the Newton polygon) and the  unit vector directed along the $q_1$--axis $\vec{e}_{q_1}$. In the case of a vertex by $\psi$ we denote the angle between the vectors $\varepsilon (1,\text{Re}\,r)$ and $\vec{e}_{q_1}$, where $\varepsilon=\pm1 $ and $r$ is the exponent in the expression $y=cx^r$. The sign of $\varepsilon$ is chosen in such a way that the vector $\varepsilon (1,\text{Re}\,r)$ lies in the domain bounded by the rays passing through the external normals of the edges attached to the vertex (excluding the rays themselves). If  $\displaystyle 0\leq \psi<\frac{\pi}{2}$, then the balance under consideration is related to the point $x=\infty$  and we obtain $x\rightarrow \infty$ for  the corresponding power asymptotics. If  $\displaystyle \frac{\pi}{2}<\psi\leq\pi$, then the balance under consideration is related to the point $x=0$  and the power asymptotics exists whenever $x\rightarrow 0$.

It is necessary to consider both normals (for the edge) and both  vectors $\pm (1,\text{Re}\,r)$ for the vertex whenever the Newton polygon degenerates to an edge or a vertex.

Non--trivial solutions of the equation $E_0[y(x),x]=0$ with $E_0[y(x),x]$ being a dominant balance
of equation  \eqref{Alg_ODE} produce asymptotics
of solutions of the original equation \cite{Bruno01, Bruno02}.
If the dominant balances related to the point $x=\infty$  are considered, then  the  corresponding asymptotics
exist provided that $x\rightarrow\infty$.
For each power solution $y(x)=b_0x^r$, $b_0\neq0$ of equation $E_0[y(x),x]=0$ one can construct a
 series satisfying  equation  \eqref{Alg_ODE} and possessing this power solution as the highest--order term.

In what follows we shall be interested in Puiseux series near the point $x=\infty$  that satisfy equation \eqref{Alg_ODE}.
In this case we should consider
only the dominant balances near the point $x=\infty$ that possess power solutions $y(x)=b_0x^{r}$ with $r\in\mathbb{Q}$.
A method of finding the Puiseux series of the
form \eqref{Puiseux_inf} that satisfy equation \eqref{Alg_ODE} can be subdivided into several steps.

At  \textit{the first step} one finds the balance $E_0[y(x),x]$ and obtains   power solutions $y(x)=b_0x^{r}$,
where $r\in\mathbb{Q}$ of the equation $E_0[y(x),x]=0$.

At \textit{the second step} one  calculates the formal  G\^{a}teaux derivative of the balance $E_0[y(x),x]$ at the solution $y(x)=b_0x^{r}$:
\begin{equation}
\begin{gathered}
 \label{Puiseux_Gatoux}
\frac{\delta E_0}{\delta y}[b_0x^{r}]=\lim_{s\rightarrow 0 }\frac{E_0[b_0x^{r}+sx^{r-j},x]-E_0[b_0x^{r},x]}{s}=V(j)x^{\tilde{r}}.
\end{gathered}
\end{equation}
In this expression $V(j)$ is a polynomial with respect to $j$ with coefficients in the field $\mathbb{C}$. The zeros of $V(j)$ are
called \textit{the Fuchs indices (or resonances)} of the
balance $E_0[y(x),x]$ and its power solution $y(x)=b_0x^{r}$.
The Fuchs indices that are positive rational numbers are crucial for further analysis. Finally, one takes all such Fuchs indices: $0<j_1<\ldots$ $<j_{m_0}$
 and calculates the number $n_0=\text{lcm}(q_1,\ldots, q_{m_0},r_2)$, where $q_m$ and $r_2$ are defined as $j_m=p_m/q_m$, $r=r_1/r_2$,
  $p_m\in \mathbb{N}$, $q_m\in \mathbb{N}$, $(p_m,q_m)=1$, $1\leq m \leq m_0$, $r_1\in \mathbb{Z}$, $r_2\in \mathbb{N}$, $(r_1,r_2)=1$.
 By lcm we denote the lowest common multiple.

At \textit{the third step} one verifies existence of the Puiseux series of the form \eqref{Puiseux_inf} with $l_0=rn_0$.
Substituting series \eqref{Puiseux_inf}
into equation  \eqref{Alg_ODE} one can find the recurrence relation for its coefficients. This relation takes  the form
\begin{equation}
\begin{gathered}
 \label{Puiseux_rec}
V\left(\frac{k}{n}\right)b_k=U_k(b_0,\ldots,b_{k-1}),\quad k\in\mathbb{N},
\end{gathered}
\end{equation}
where $U_k$ is a polynomial of its arguments. Note that $U_k$, as well as $V(j)$, can also depend on the parameters (if any) of the original equation.
The equations $U_{nj_m}=0$, $1\leq m\leq m_0$ are called \textit{compatibility conditions}. If at least one of the   compatibility conditions is not satisfied,
then the Puiseux series under consideration does not exist. If all the compatibility conditions are satisfied, then the corresponding  Puiseux series exists and possesses
 arbitrary coefficients: $b_{nj_{m_0}}$ and, possibly, all or some  of the coefficients $b_0$,  $b_{nj_{1}}$, $\ldots$, $b_{nj_{m_0-1}}$. Note that the
 Puiseux series under consideration possesses uniquely determined coefficients whenever there are no non--negative rational Fuchs indices.

\textit{Remark 1.} If one wishes to find all the Puiseux series of the form \eqref{Puiseux_inf} that satisfy the original equation, then
it is necessary to implement the procedure described above for all the dominant balances near the point $x=\infty$ and for all their power solutions.

\textit{Remark 2.} The similar method is applicable to an arbitrary--dimensional  polynomial dynamical system. In relation to a dynamical system the Fuchs indices are
called \textit{the  Kovalevskaya exponents}.

\textit{Remark 3.} If at least one compatibility condition is not satisfied, then it is still possible to find the corresponding series, but in this
case there exist coefficients $\{b_k\}$ that are no longer constants, instead, they are the functions of $\log x$.
Note that such power--logarithmic series $y(x)$ do not satisfy polynomial equations
of the form $F(x,y)=0$, $F(x,y)\in\mathbb{C}[x,y]$.

The equation of the form \eqref{ODE_y} related to dynamical system  \eqref{FH1_DS} with $e=0$  is the following
\begin{equation}\label{FH_ODE}
\{-x^3+\sigma x-y+\delta\}y_x-\{\alpha x+\beta y\}=0.
\end{equation}
The Newton polygon of this equation is presented in figure \ref{f:FH_polygon}.

\begin{figure}
 \centerline{\epsfig{file=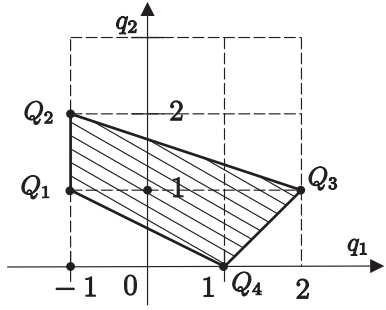,width=50mm}}
 \caption{The Newton polygon of equation \eqref{FH_ODE} with $\alpha\neq0$ and $\delta\neq0$.}\label{f:FH_polygon}
\end{figure}

The balances giving power asymptotics at $x\rightarrow \infty$ and their  power solutions take the form
\begin{equation}\label{FH_Balances}
\begin{gathered}
(Q_2,Q_3):\quad (x^3+y)y_x=0,\quad y(x)=-x^3;\hfill\\
Q_3:\quad \qquad \,\, x^3y_x=0,\quad y(x)=a_0;\hfill\\
(Q_3,Q_4):\quad x^3y_x+\alpha x=0,\quad y(x)=\frac{\alpha}{x}.\hfill
\end{gathered}
\end{equation}
In these expressions $a_0\neq 0$ is an arbitrary constant. The corresponding series turn out to be Laurent series. They are the following
\begin{equation}\label{FH_Series}
\begin{gathered}
(I):\quad y(x)=-x^3+\left(\sigma-\frac{\beta}{3}\right)x+\delta+\ldots;\hfill\\
(II):\quad y(x)=a_0+\frac{\alpha}{x}+\frac{\beta a_0}{2x^2}+\frac{\alpha(\sigma+\beta)}{3x^3}\ldots;\hfill\\
(III):\quad y(x)=\frac{\alpha}{x}+\frac{\alpha(\sigma+\beta)}{3x^3}+\ldots.\hfill
\end{gathered}
\end{equation}

Series $(I)$ and $(III)$ have uniquely determined coefficients, while series $(II)$ has one arbitrary coefficient $a_0$. Further, series $(III)$ is a partial case of series $(II)$.
Indeed, setting  $a_0=0$ in $(II)$, we get $(III)$. If $\alpha=0$ the edge $(Q_3, Q_4)$ disappears  and we do not have series $(III)$.

\begin{figure}
 \centerline{\epsfig{file=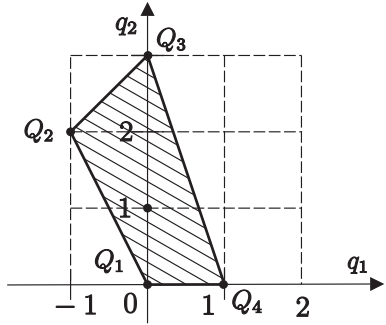,width=50mm}}
 \caption{The Newton polygon of equation \eqref{FH_ODE_x} with $\alpha\neq0$ and $\delta\neq0$.}\label{f:FH_polygon_x}
\end{figure}

Unfortunately, we are in situation with infinite number of admissible Puiseux series near the point $x=\infty$. Further, let the variables $x$ and $y$
change their roles. Now we suppose that $x=x(y)$. The equation of the form \eqref{ODE_y} (with $x\leftrightarrow y$) can be written as
\begin{equation}\label{FH_ODE_x}
\{-x^3+\sigma x-y+\delta\}-\{\alpha x+\beta y\}x_y=0.
\end{equation}
Its Newton polygon is given in figure \ref{f:FH_polygon_x}. Analyzing the Newton polygon of figure \ref{f:FH_polygon_x}, we see that there exists only one balance
 giving power asymptotics at $y\rightarrow \infty$. Indeed, other balances related to the point $y=\infty$ do not have power solutions.
 This balance and its power solutions are the following
\begin{equation}\label{FH_Balances_x}
\begin{gathered}
(Q_3,Q_4):\quad x^3+y=0,\quad x^{(j)}(y)=b_0^{(j)}y^{1/3},\quad b_0^{(j)}=\{(-1)^{1/3}\}_j, \quad j =1,2,3.\hfill
\end{gathered}
\end{equation}
Here $b_0$ is one of the cubic roots of $-1$.  We obtain three distinct  Puiseux series. They take the form
\begin{equation}\label{FH_Series_x}
\begin{gathered}
 x^{(j)}(y)=b_0^{(j)}y^{1/3}+\frac19\left(b_0^{(j)}\right)^2(\beta-3\sigma)y^{-1/3}+\ldots, \quad j =1,2,3.
\end{gathered}
\end{equation}

Now we are in situation with finite number of admissible Puiseux series. We shall not construct admissible Puiseux series near the origin ($x=0$ and $y=0$). Irreducible invariant algebraic curves  of dynamical system \eqref{FH1_DS} will be classified in the next section.

\section{Invariant algebraic curves of  the FitzHugh--Nagumo system} \label{FH_Darboux}

Let $F(x,y)$ be an  invariant algebraic curve  of dynamical system \eqref{FH1_DS} with $e=0$.
Then $F(x,y)$ satisfies the equation
\begin{equation}\label{FH_DP}
\{-x^3+\sigma x-y+\delta\}F_x+\{\alpha x+\beta y\}F_y=\lambda(x,y) F.
\end{equation}

\begin{lemma}\label{L:FH_cof}
If $F(x,y)$ is an invariant algebraic curve of system \eqref{FH1_DS}, then
\begin{equation}\label{FH_DP_LO}
F(x,y)=\mu_0y^N+\sum_{k=0}^{N-1}c_k(x)y^k,\quad N\in\mathbb{N}
\end{equation}
and its cofactor $\lambda$ is $\lambda=A_2x^2+A_1x+A_0$, where $A_2=-M$
with $M$ being the degree of $F(x,y)$ with respect to~$x$.
\end{lemma}

\begin{proof}
By direct calculations we find that there are no invariant algebraic curves that do not depend on $y$.

Let $F$ and $\lambda$ have degrees $N\in\mathbb{N}$ and $l\in\mathbb{N}_0$ with respect to $y$ accordingly. Substituting relations $F=\mu(x)y^N$, $\mu(x)\not\equiv0$ and $\lambda=\lambda_0(x)y^l$ with $\mu(x)$, $\lambda_0(x) \in\mathbb{C}[x]$
into equation \eqref{FH_DP} and balancing higher--order terms in $y$ yields $l=1$ and $\mu_x=-\lambda_0(x)\mu$. Since $\mu(x)$ is a polynomial, we get $\lambda_0(x)=0$, $\mu(x)=\mu_0$, where $\mu_0$ is a constant.  In addition we see that the cofactor $\lambda$ does not depend on $y$.

Now suppose that $F$ and $\lambda$ have degrees $M\in\mathbb{N}_0$ and $s\in\mathbb{N}_0$ with respect to $x$ accordingly. Taking expressions $F=\nu(y)x^M$, $\nu(y)\not\equiv0$ and $\lambda=A_sx^s$ with $\nu(y) \in\mathbb{C}[y]$, $A_s\in\mathbb{C}$ and arguing as above, we get $s=2$, $A_2=-M$. This completes the proof.

\end{proof}

In what follows we shall assume without loss of generality that $\mu_0=1$.

\begin{table}
       \begin{tabular}[pos]{l l l}
        \hline
        \textit{Invariant algebraic curves} & \textit{Cofactors} & \textit{Parameters}\\
        \hline
        $ $ & $ $\\
          $ y$ & $\beta$ & $ \alpha=0$\\
        $ y+x^3+(\frac{\beta}{3}-\sigma) x$ & $-3x^2+\sigma+\frac{2\beta}{3}$ & $ \alpha=\frac{2\beta}{9}(\beta-3\sigma) $, $\delta=0$\\
        $ y^2+(x^3+2\beta x)y-\frac{\beta^2}{3} x^2-\frac{16\beta^3}{27}$ & $\displaystyle-3x^2$ & $ \alpha=\frac{\beta^2}{3}$, $ \sigma=-\frac{5\beta }{3}$, $\delta=0$\\
        $ y^2+(x^3-\frac{2\beta}{3}x)y+\frac{\beta^2}{9}x^2$ & $-3x^2+\frac{8\beta}{3}$ & $ \alpha=-\frac{\beta^2}{9}$, $\sigma=\beta$, $\delta=0$\\
        $ y^2+\left(x^3-\frac{10\beta}{27}x-\frac{32\beta^{3/2}}{729}\right)y+\frac{64\beta^{3/2}}{729}x^3$ & $-3x^2+\frac{64\beta}{27}$ & $ \alpha=-\frac{5\beta^2}{81}$, $\sigma=\frac{19\beta}{27}$, $\delta=\frac{32\beta^{3/2}}{243}$\\
       $ \qquad +\frac{5\beta^{2}}{81}x^2 +\frac{224\beta^{5/2}}{19683}x +\frac{112\beta^{3}}{177147}$ & $ $ & $ $\\
       $ y^2+\left(x^3-\frac{10\beta}{27}x+\frac{32\beta^{3/2}}{729}\right)y-\frac{64\beta^{3/2}}{729}x^3$ & $-3x^2+\frac{64\beta}{27}$ & $ \alpha=-\frac{5\beta^2}{81}$, $\sigma=\frac{19\beta}{27}$, $\delta=-\frac{32\beta^{3/2}}{243}$\\
       $ \qquad +\frac{5\beta^{2}}{81}x^2 -\frac{224\beta^{5/2}}{19683}x +\frac{112\beta^{3}}{177147}$ & $ $ & $ $\\
       $ y^3+(x^3+\frac{56\beta}{27}x)y^2-(\frac{32\beta^2}{81}x^2-\frac{400000\beta^3}{531441})y$ & $-3x^2+\frac{25\beta}{27}$ & $ \alpha=\frac{16\beta^2}{81}$, $\sigma=-\frac{47\beta}{27}$, $\delta=0$\\
       $ \qquad -\frac{16384\beta^3}{531441}x^3-\frac{800000\beta^4}{14348907}x$ & $ $ & $ $\\
       $ y^3+(x^3-\frac{7\beta}{6}x)y^2+\frac{2\beta^2}{9}x^2y+\frac{2\beta^3}{27}x^3$ & $-3x^2+\frac{25\beta}{6}$ & $ \alpha=-\frac{\beta^2}{9}$, $\sigma=\frac{3\beta}{2}$, $\delta=0$\\
       $y^3+(x^3-\frac{4\beta}{7}x)y^2+\frac{16\beta^2}{147}x^2y-\frac{64\beta^3}{9261}x^3$ & $-3x^2+\frac{25\beta}{7}$ & $ \alpha=-\frac{8\beta^2}{147}$, $\sigma=\frac{19\beta}{21}$, $\delta=0$\\
        $ $ & $ $\\
                              \hline
    \end{tabular}
    \caption{Irreducible invariant algebraic curves of FitzHugh--Nagumo dynamical system \eqref{FH1_DS} at  $e=0$ and $|\alpha|+|\beta |>0$.} \label{T:FH_inv}
\end{table}

\textit{Proof of  Theorem \ref{T:FH_DP}.} Suppose that  $F(x,y)$ is an irreducible invariant algebraic curve  of FitzHugh--Nagumo dynamical system \eqref{FH1_DS} with $e=0$. In view of  Theorems \ref{T2}, \ref{T3}, Lemma \ref{L:FH_cof},  and results of the previous section we get the following representations in the fields $\mathbb{C}_{\infty}\{x\}$, $\mathbb{C}_{\infty}\{y\}$:
\begin{equation}
\begin{gathered}
 \label{FH_F_rep}
\mathbb{C}_{\infty}\{x\}:\quad F(x,y)=\left\{y+x^3-\left(\sigma-\frac{\beta}{3}\right)x-\delta-\ldots\right\}^{n}\prod_{j=1}^{m}\{y-a_0^{(j)}-\ldots\};
\end{gathered}
\end{equation}
\begin{equation}
\begin{gathered}
 \label{FH_F_rep_x}
\mathbb{C}_{\infty}\{y\}:\quad F(x,y)=\nu(y)\prod_{j=1}^{3}\left\{x-b_0^{(j)}y^{1/3}-\ldots\right\}^{k_j},
\end{gathered}
\end{equation}
where $n=0$ or $n=1$, $m\in\mathbb{N}_0$, $k_j=0$ or $k_j=1$. In addition we suppose that the last product in \eqref{FH_F_rep} is unite whenever $m=0$.  Using Theorem \ref{T2}, we see that the degree of $F(x,y)$ with respect to $x$ is at most $3$. Moreover, it follows from  representations \eqref{FH_F_rep} and \eqref{FH_F_rep_x} that either $M=0$  or $M=3$.

First we consider the case $M=0$. From expressions  \eqref{FH_F_rep}, \eqref{FH_F_rep_x}, and Lemma \ref{L:FH_cof} we obtain that $F(x,y)$ does not depend on $x$ and $n=0$, $A_2=0$. Substituting $F(x,y)=F(y)\in\mathbb{C}[y]$ into equation \eqref{FH_DP} and setting to zero the coefficients at $x^1$ and $x^0$, we find $A_1=0$, $A_0=\beta$, $\alpha=0$, and $F(x,y)=y$.

Secondly we consider the case $M=3$. As a result we obtain $n=1$, $A_2=-3$. Substituting $F(x,y)=c_3(y)x^3+c_2(y)x^2+c_1(y)x+c_0(y)$ into equation \eqref{FH_DP} and setting to zero the coefficients at $x^4$, $x^3$ and $x^2$, we express $c_2(y)$, $c_1(y)$, and $c_0(y)$ via $c_3(y)$ and its derivatives. The result is
\begin{equation}
\begin{gathered}
 \label{FH_F_rep_x_c}
c_2=-\alpha c_{3,y}+A_1c_3,\quad c_1=\frac{\alpha^2}{2}c_{3,yy}+\left(\frac{\beta}{2}y+\alpha A_1\right)c_{3,y}+\frac12\left(A_0+ A_1^2-3\sigma\right)c_{3,y},\\
c_0=-\frac{\alpha^3}{6}c_{3,yyy}+\frac{\alpha}{2}\left(\beta y+\alpha A_1\right)c_{3,yy}-\left\{\frac12\beta A_1y-\frac{\alpha}{6}\left(7\sigma
+\beta-3A_0-3A_1^2\right)\right\}c_{3,y}\\
+\left\{y-\delta+\frac16A_1\left(A_1^2+3A_0-7\sigma\right)\right\}c_3.\hfill
\end{gathered}
\end{equation}
Note that these relations are linear in $c_3(y)$ and its derivatives.  Further, setting to zero the coefficients at $x^1$ and $x^0$ and using relations for $c_2(y)$, $c_1(y)$, and $c_0(y)$, we obtain two fourth--order linear ordinary differential equations for the polynomial $c_3(y)$. We use the method of undetermined coefficients to find their polynomial solutions of degrees $0$, $1$, $2$. The results are given in table~\ref{T:FH_inv}.

Finally, we suppose that there exists a polynomial solution of degree $K\geq3$. We make the substitution $c_3(y)=y^K+u_{K-1}y^{K-1}+\ldots $ and set to zero the coefficients at $y^{K+1-j}$, $j\geq0$. The algebraic equations with $j=0$ give the values of $A_1$ and $A_0$:
\begin{equation}
\begin{gathered}
 \label{FH_F_Cof}
A_1=0,\quad A_0=\left(K+\frac23\right)\beta+\sigma.
\end{gathered}
\end{equation}
Solving other algebraic equations, we find the values of $u_{K-j}$, $j\geq1$ and necessary conditions for such a polynomial solution to exist.
 For example, the equations at $j=1$ take the form
 \begin{equation}
\begin{gathered}
 \label{FH_F_Cof1}
2\beta(3\sigma-\beta)+9\alpha(3K+1)=0,\quad \beta(2\delta-3u_{K-1})=0.
\end{gathered}
\end{equation}
 Taking eight equations ($1\leq j\leq 8$), we see that this system is inconsistent if we require that the resulting algebraic curve is irreducible. Recall that the case $\alpha=0$ and $\beta=0$ is simple (see section \ref{Introduction}) and we exclude it here. This completes the proof.

\section{Conclusion}

In this article we have studied a relationship between   invariant algebraic curves of a polynomial dynamical system and the  Puiseux series satisfying an ordinary differential equation corresponding to the system. A bound on the degrees of irreducible invariant algebraic curves  for a wide class of  polynomial vector fields has been derived. It is shown that the structure of the Puiseux series near infinity can be used to find all irreducible algebraic curves explicitly.  Using this approach we have classified irreducible invariant algebraic curves of  the FitzHugh--Nagumo system. For further developments of the method of the present article see \cite{Demina06, Demina07}.



\bibliography{ref}

\end{document}